\documentclass{amsart}
\usepackage[latin1]{inputenc}
\usepackage[all]{xy}
\usepackage{amssymb, amsmath, amscd, geometry}
\usepackage{graphicx}
\usepackage{setspace}

\numberwithin{equation}{section}
\input xy
\xyoption{all}

\usepackage{latexsym}
\usepackage[usenames]{color}
\usepackage{epic} 
\usepackage{mathrsfs}
\usepackage{euscript}
\usepackage{rotating}
\usepackage{verbatim}

\newtheorem{lemma}{Lemma}[section]
\newtheorem{theorem}[lemma]{Theorem}

\newtheorem*{theorem*}{Theorem}

\theoremstyle{definition}

\newtheorem{remark}{Remark}[section]

      \newcommand{\N}{{\mathbb N}}
      \newcommand{\R}{{\mathbb R}}
      \newcommand{\C}{{\mathbb C}}

\newcommand{\sign}{\operatorname{sign}}

\newcommand{\uno}{1\!\!1}

\title[On the largest Bell violation attainable by a quantum state]{On the largest Bell violation attainable by a quantum state}

\author{Carlos Palazuelos}

\addtocounter{tocdepth}{-1}

\begin{document}

\addtolength{\parskip}{+1ex}



\maketitle

\begin{abstract}
We study the projective tensor norm as a measure of the largest Bell violation of a quantum state. In order to do this, we consider a truncated version of a well-known SDP relaxation for the quantum value of a two-prover one-round game, one which has extra restrictions on the dimension of the SDP solutions. Our main result provides a quite accurate upper bound for the distance between the classical value of a Bell inequality and the corresponding value of the relaxation. Along the way, we give a simple proof that the best complementation constant of $\ell_2^n$ in $\ell_1(\ell_\infty)$ is of order $\sqrt{\ln n}$. As a direct consequence, we show that we cannot remove a logarithmic factor when we are computing the largest Bell violation attainable by the maximally entangled state.
\end{abstract}

\section*{Introduction}
A standard scenario to study quantum nonlocality consists of two spatially separated and non-communicating parties, usually called Alice and Bob. Each of them can choose among different measurements, labeled by $x=1,\cdots , N$ in the case of Alice and $y=1,\cdots , N$ in the
case of Bob. The possible outcomes of these measurements are labeled by $a=1,\cdots , K$ in the case of Alice and $b=1,\cdots
, K$ in the case of Bob. Following the standard notation, we will refer the observables $x$ and $y$ as \emph{inputs} and call $a$
and $b$ \emph{outputs}. For fixed $x,y$, we will consider the probability distribution $(P(a,b|x,y))_{a,b=1}^K$ of positive real numbers satisfying $$\sum_{a,b=1}^KP(ab|xy)= 1.$$ The collection $P=\Big(P(a,b|x,y)\Big)_{x,y; a,b=1}^{N,K}$ will be also referred as a \emph{probability distribution}.

Given a probability distribution $P=\big(P(a,b|x,y)\big)_{x,y; a,b=1}^{N,K}$, we will say that $P$ is
\begin{enumerate}
\item [a)] \emph{Classical} if
\begin{equation*}\label{classical}
P(a,b|x,y)=\int_\Omega P_\omega(a|x)Q_\omega(b|y)d\mathbb{P}(\omega)
\end{equation*}
for every $x,y,a,b$, where $(\Omega,\Sigma,\mathbb{P})$ is a probability space, $P_\omega(a|x)\ge 0$ for all $a,x,\omega$, $\sum_a
P_\omega(a|x)=1$ for all $x,\omega$ and analogous conditions for the $Q_\omega(b|y)$'s. We denote the set of classical probability
distributions by $\mathcal{L}$.
\item [b)] \emph{Quantum} if there exist two Hilbert spaces $H_1$ and $H_2$ such that
\begin{equation*}\label{quantum}
P(a,b|x,y)=tr(E_x^a\otimes F_y^b \rho)
\end{equation*}
for every $x,y,a,b$, where  $\rho\in B(H_1\otimes H_2)$ is  a density operator acting on $H_1\otimes H_2$ and $(E_x^a)_{x,a}\subset B(H_1)$, $(F_y^b)_{y,b}\subset B(H_2)$ are two sets of operators representing POVM measurements on Alice's and Bob's systems. That is, $E_x^a\geq 0$ for every  $x,a$, $\sum
_{a}E_x^a=\uno$ for every $x$, and analogous conditions for the $F_y^b$'s. We denote the set of
quantum probability distributions by $\mathcal{Q}$.
\end{enumerate}
It is not difficult to see that both $\mathcal{L}$ and $\mathcal{Q}$ are convex sets verifying $\mathcal{L}\subseteq \mathcal{Q}$. Furthermore, $\mathcal{L}$ is a polytope.
Note that in order to talk about $\mathcal L$ and $\mathcal Q$ we
must fixed the number of inputs $N$ and outputs $K$ in the previously introduced Alice-Bob
scenario. However, we will just write $P\in \mathcal L$ (resp. $Q\in
\mathcal Q$) when $N$ and $K$ are clear from the context. Given an element $M=(M_{x,y}^{a,b})_{x,y;a,b=1}^{N,K}\in \R^{N^2K^2}$, we denote
$$\langle
M,P\rangle=\sum_{x,y;a,b=1}^{N,K}M_{x,y}^{a,b}P(a,b|x,y)$$for every
probability distribution $P=(P(a,b|x,y))_{x,y; a,b=1}^{N,K}$. Then,
we define the \emph{largest Bell violation of $M\in \R^{N^2K^2}$} by
\begin{align*}
LV(M)=\frac{\omega^* (M)}{\omega (M)},
\end{align*}where $\omega^* (M):=\sup\Big\{|\langle M,Q\rangle|: Q\in \mathcal{Q} \Big\}$ and $\omega (M):=\sup\Big\{|\langle M,P\rangle|: P\in\mathcal{L}\Big\}$ (see \cite{JP}, \cite{JPPVW2}, \cite{JPPVW}). Actually, we must restrict this definition to those elements $M$ which do not vanish on all $\mathcal L$. In the following we will assume this fact and we will write $M\in \mathcal M^{N,K}$. Any $M\in \mathcal M^{N,K}$ will be referred as a \emph{Bell inequality}\footnote{Formally, Bell inequalities are those inequalities which describe the facets of the set $\mathcal{L}$. However, for our purposes it is very convenient considering this more general definition.}. We talk about a \emph{Bell inequality violation} when we have $LV(M)> 1$ for some $M\in \mathcal M^{N,K}$ (see \cite{Tsirelson}). Note that this fact is equivalent to say that $\mathcal{L}$ is strictly contained in $\mathcal{Q}$. This is also referred as \emph{quantum nonlocality}.

Quantum nonlocality is a crucial point in many different areas of
quantum information and quantum computation. Some examples can be
found in quantum cryptography (\cite{ABGMPS}, \cite{AMG}), in
testing random numbers (\cite{PCMBM}), in complexity theory
(\cite{CHTW}, \cite{KRT}) and in communication
complexity (\cite{BCMW}). This has motivated an increased interest
in the study of the value $LV(M)$ for some fixed $M$'s and also in
the study of $\sup_M LV(M)$, as a way of \emph{quantifying quantum
nonlocality} (see \cite{BRSW}, \cite{JP}, \cite{JPPVW2}, \cite{JPPVW}, \cite{PWJPV} for some recent works on the topic). In this work we
will be also concerned with quantifying quantum nonlocality but we
will change our perspective. Here, we will focus on the
quantum states. Our main question is: \emph{Given an $n$-dimensional bipartite state $\rho$, how large can its Bell violations be?} The mathematical formulation of this question requires some extra notation. We will denote by $\mathcal Q_\rho$ the set of all quantum probability distributions constructed with the state $\rho$ and,
given $M\in \mathcal M^{N,K}$, we will denote
$$\omega^*_\rho(M)=\sup\Big\{\big|\langle M, Q\rangle\big|:Q\in Q_\rho\Big\}\text{    }\text{  and   }\text{    }LV_\rho(M)=\frac{\omega_\rho^*(M)}{\omega(M)}.$$
Finally, we will define our key object:
\begin{align*}
LV_\rho:=\sup_{N,K}\sup_{M\in \mathcal M^{N,K}}LV_\rho(M).
\end{align*}When we are dealing with pure states $\rho=|\psi\rangle\langle \psi|$\footnote{Here, $|\psi\rangle$ denotes a unit vector in a Hilbert space and  $|\psi\rangle\langle \psi|$ denotes the rank-one projection defined by this element. We will explain the ket-bra notation in Section \ref{Section: Projective}.} we will just write $LV_{|\psi\rangle}$.
Then, the previous question can be reformulated as: \text{   }\emph{How large can} $LV_\rho$ \emph{be? }

The quantity $LV_\rho$ was first considered in \cite{JP} as a natural measure of \emph{how nonlocal a quantum state $\rho$ is}. Indeed, since nonlocality refers to probability distributions, it is natural to quantify the nonlocality of a state $\rho$ by measuring how nonlocal the quantum probability distributions constructed with $\rho$ can be. $LV_\rho$ measures exactly this.
%
%
In many cases one is interested in studying the value $LV_\rho(M)$
for fixed $M$ and $\rho$ and also in $\sup_\rho LV_\rho(M)=LV(M)$
for a fixed Bell inequality $M\in \mathcal M^{N,K}$. In this context
one can find very interesting works that mainly deal with particular
Bell inequalities like CHSH (\cite{CHSH}), CGLMP (\cite{CGLMP}) and
$I_{3322}$ (\cite{Fro}). Recent works have treated this problem
from a more general point of view by studying the asymptotic
behavior of $\sup_{M\in \mathcal M^{N,K}}LV(M)$ for fixed $N$ and
$K$ (\cite{BRSW}, \cite{JP}, \cite{JPPVW2}, \cite{JPPVW}, \cite{PWJPV}). Note that in these problems we fix the number of inputs $N$ and
outputs $K$, whereas the dimension $n$ of the state (and operator
measurements) is a free parameter in the optimization. In contrast,
$n$ is the fixed parameter in the problem considered in this work
(since we fix our state $\rho$), whereas we must consider $N$ and
$K$ as free parameters in order to optimize over all Bell
inequalities and all quantum measurements. This means that the
problem considered here is, somehow, dual of those considered
before. Our first result relates the quantity $LV_\rho$ with the projective tensor norm of $\rho$ (see definition in Section \ref{Section: Projective}), already
used in several contexts of quantum information theory (see for
instance \cite{HaMo}, \cite{Szarek}, \cite{Rudolph}).
\begin{theorem}\label{Theorem I}
Given an $n$-dimensional bipartite quantum state $\rho$, we can realize it as an element in the algebraic tensor product $M_n\otimes M_n$\footnote{The fact that $\rho$ is a state implies that $\|\rho\|_{S_1^{n^2}}=1$, where here $S_1^{n^2}$ denotes the space $M_{n^2}$ endowed with the trace norm.}, where $M_n$ denotes the space of complex matrices of size $n$. Then, 
\begin{align*}
LV_\rho\leq K\|\rho\|_{S_1^n\otimes_\pi S_1^n},
\end{align*}where $S_1^n$ denotes the space $M_n$ endowed with the trace norm, $\pi$ denotes the projective tensor norm and $K$ is a universal constant independent of the dimension\footnote{$K$ can be taken lower than or equal to $4$.}. In particular,
\begin{align*}
\sup_\rho LV_\rho\leq K n,
\end{align*}where the supremum runs over all $n$-dimensional bipartite states.
\end{theorem}
Hence, though the definition of $LV_\rho$ involves the supremum over all $N,K\in \N$ and $M\in \mathcal M^{N,K}$, we have that $LV_\rho< \infty$ for every finite dimensional state $\rho$.  Interestingly, in the recent paper  \cite{BRSW} the authors showed that the upper bound $O(n)$ given in Theorem \ref{Theorem I} is very tight. 
\begin{theorem}(Buhrman, Regev, Scarpa, de Wolf, \cite{BRSW})\label{Theorem BRSW I} 
Let us consider the $n$-dimensional maximally entangled state  $\rho:=|\psi_n\rangle\langle\psi_n|$, with $|\psi_n\rangle:= \frac{1}{\sqrt{n}}\sum_{i=1}^ne_i\otimes e_i\in \ell_2^n\otimes_2\ell_2^n$. Then,
\begin{align*}
LV_{|\psi_n\rangle}\geq C \frac{n}{(\ln n)^2},
\end{align*}
where $C$ is a universal constant independent of the dimension.
\end{theorem}
As we will explain in Section \ref{lower bounds for every state}, a careful study of this result allows us to show that for every pure state $\rho$ in dimension $n$ we have
\begin{align}\label{Estimate extension BRSW}
LV_{\rho}\geq C'\frac{\|\rho\|_{S_1^n\otimes_\pi S_1^n}}{(\ln n)^2},
\end{align}where $C'$ is a universal constant independent of the dimension\footnote{We will explain in Section \ref{lower bounds for every state} that one can actually obtain better lower bounds for  $LV_{\rho}$.}.

Theorem \ref{Theorem I} and Theorem \ref{Theorem BRSW I} (together with Equation (\ref{Estimate extension BRSW})) show the projective tensor norm as a good candidate to measure the largest Bell violation attainable by a (pure) state. This reminds us Rudolph's characterization of entangled states:  \emph{A quantum state $\rho\in M_n\otimes M_n$ is entangled if and only if $\|\rho\|_{S_1^n\otimes_\pi S_1^n}>1$}. In this sense the previous estimates show a link between quantum entanglement and quantum nonlocality, contrary to the spirit of the most recent results on the topic (see for instance \cite{BDHSW}, \cite{JP}, \cite{MeSc}). Then, we can wonder whether the projective tensor norm of a state $\|\rho\|_{S_1^n\otimes_\pi S_1^n}$ can measure its largest Bell violation $LV_\rho$ up to, maybe, a constant factor. Unfortunately, we will show in this work that this is not the case.
\begin{theorem}\label{Main I}
Let $\rho=|\psi_n\rangle\langle\psi_n|$ be the maximally entangled state in dimension $n$, then
\begin{align*}
LV_{|\psi_n\rangle}\leq D \frac{n}{\sqrt{\ln n}}
\end{align*}for a certain universal constant $D$.
\end{theorem}
Theorem \ref{Main I} partially answers the open question
posed in \cite[Section 1.3]{BRSW} about the possibility of removing
the logarithmic factor in Theorem \ref{Theorem BRSW I}. One cannot do
this if we restrict to the $n$-dimensional maximally entangled state (as it was used in \cite{BRSW}). On the other
hand, since $\||\psi_n\rangle\langle\psi_n|\|_{S_1^n\otimes_\pi
S_1^n}=n$, we deduce that we can not use the projective tensor norm
as an ``accurate'' measure  of the largest Bell violation attainable
by a quantum state. We have to consider, in general, an extra
logarithmic factor.  We must also mention that, beyond their own interest, these logarithmic-like estimates are very useful to obtain non-multiplicative results. Indeed, in the very recent paper \cite{Palazuelos Super-act} the previous estimates have been used to show that the largest Bell violation of a state $LV_\rho$ is a highly non-multiplicative measure. Actually, similar techniques have been used to show that quantum nonlocality can be superactivated (\cite{Palazuelos Super-act}).

Theorem \ref{Main I} is a consequence of a stronger result. Given any Bell inequality $M\in \mathcal M^{N,K}$, let us consider the following
optimization problem, which optimizes over families of real $n$-dimensional vectors $\{u_x^a\}_{x,a=1}^{N,K}$,
$\{v_y^b\}_{y,b=1}^{N,K}$, $z$:
\begin{align}\label{restrictions}
\begin{array}{lll}
                 \overline{\omega}_{OP_n}(M):= & \text{\textbf{Maximize:}} & \Big|\sum_{x,y,a,b=1}^{N,K}M_{x,y}^{a,b}\langle u_x^a, v_y^b\rangle\Big|\\
                    &  \text{\textbf{Subject to:}}   & \forall x,y, \sum_{a=1}^K u_x^a=\sum_{b=1}^K v_y^b=z \text{   }  (*),\\
                     &            & \forall x, \sup_{\alpha_x^a=\pm 1}\big\|\sum_{a=1}^K\alpha_x^a u_x^a\big\|\leq 1, \\
                 &            & \forall y, \sup_{\alpha_y^b=\pm 1}\big\|\sum_{b=1}^K\alpha_y^b v_y^b\big\|\leq 1.
                           \end{array}\end{align}
As we will explain in Section \ref{Section: Relaxation},  $\overline{\omega}_{OP_n}(M)$ is a natural generalization of a well-known semidefinite programming (SDP) relaxation for the classical and quantum value of a two-prover one-round game\footnote{In fact, restriction (*) is not needed when we consider two-prover one-round games, but it must be considered when we work with general Bell inequalities (see Section  \ref{Section: Relaxation} for a complete explanation).}. Then, the main result of this work states as follows.
\begin{theorem}\label{Main II}
For all natural numbers $n$, $N$, $K$ and every $M\in \mathcal M^{N,K}$ we have
\begin{align*}
\overline{\omega}_{OP_n}(M)\leq D'\frac{n}{\sqrt{\ln n}}\omega(M),
\end{align*}where $D'$ is a universal constant.
\end{theorem}
We think that Theorem \ref{Main II}  can be of independent interest for computer scientists. $\overline{\omega}_{OP_n}$ is the natural generalization of $SDP$ (see Section \ref{Section: Relaxation}) when we want to impose ``low dimensional solutions''(where orthogonality restrictions no longer make sense since we will have $K> n$). As far as we know the question of rounding low-dimensional solutions of these kinds of optimization problems has not received much attention. Some interesting papers in this direction are
\cite{AvZw}, \cite{BFV}, \cite{BFVII}.

Finally, it is interesting to mention that Theorem \ref{Main II} is closely related to the problem of finding the best complementation constant of $\ell_2^n$ in $\ell_1(\ell_\infty)$ (see Section \ref{Some comments about the optimality} for details). In fact, in this work we present a simple proof that such a complementation constant is of order $\Omega(\sqrt{\ln n})$ (see Theorem \ref{optimal complementation}, Part 1. below). This estimate was first proved by Bourgain in \cite{Bourgain} (see also \cite{BCLT} for an alternative proof of the order $\Omega(\ln n)^\beta$ for a certain $\beta>0$). However, the proof given in our work is completely different and arguably simpler, based on the concentration of measure phenomenon. Combined with previous results this estimate allows us to state the following result.
\begin{theorem}\label{optimal complementation}
\

\begin{enumerate}
\item[1.] Given linear maps $S:\ell_2^n\rightarrow \ell_1(\ell_\infty)$ and $T:\ell_1(\ell_\infty)\rightarrow \ell_2^n$ such that $T\circ S=id_{\ell_2^n}$. We have that $\|T\|\|S\|\geq K\sqrt{\ln n}$, where $K$ is a universal constant. Furthermore, this estimate is optimal: There exist linear maps $j:\ell_2^n\longrightarrow \ell_1^n(\ell_\infty^n)$ and $P:\ell_1^n(\ell_\infty^n)\rightarrow \ell_2^n$ such that $P\circ j=id_{\ell_2^n}$ and $\|j\|\|P\|\leq \tilde{K}\sqrt{\ln n}$, where $\tilde{K}$ is a universal constant.
\item[2.] If we consider complex Banach spaces, the same estimate holds and one has that such an optimality is true even in the following non-commutative sense: There exist linear maps $j:R_n\cap C_n\longrightarrow \ell_1(\ell_\infty)$ and $P:\ell_1(\ell_\infty)\rightarrow R_n\cap C_n$ such that $P\circ j=id_{\ell_2^n}$ and $\|j\|_{cb}\|P\|_{cb}\leq \tilde{K}\sqrt{\ln n}$, where $\tilde{K}$ is a universal constant. Here $R_n\cap C_n$ denotes the complex space $\ell_2^n$ endowed with the $R\cap C$ operator space structure, $\ell_1(\ell_\infty)$ is considered with its natural operator space structure and $\|\cdot\|_{cb}$ denotes the completely bounded norm.
\end{enumerate}
\end{theorem}
The paper is organized as follows. We start Section \ref{Section:
Projective} by giving a very brief introduction about some basic notation in quantum information theory. Then, we show that the projective tensor norm can be seen as a
good measure for the largest Bell violation of a quantum state
$LV_\rho$. In particular, we provide upper and lower bounds for this
largest Bell violation by proving Theorem \ref{Theorem I} and the estimate in (\ref{Estimate extension BRSW}). This section should be considered as the motivation for the subsequent results in the paper. In Section \ref{Section: Relaxation}
we introduce a modified version of a SDP relaxation already used in
computer sciences to approximate the classical and quantum value of
a $2P1R$-game and we explain how it is related to the quantity $\overline{\omega}_{OP_n}(M)$ introduced above. At the end of this section, we show how to obtain Theorem  \ref{Main I} from Theorem \ref{Main II}  although we postpone the proof of Theorem \ref{Main II} to Section \ref{Section: proofs}. Finally, in Section \ref{Section: proofs} we present the proof of our main result, Theorem \ref{Main II}, and we comment some points about its optimality.
\section{The projective tensor norm as a measure of the largest Bell violation}\label{Section: Projective}
Let us start this section with a brief introduction about the ket-bra notation commonly used in quantum information theory (QIT). We will denote the $n$-dimensional complex Hilbert space by $\ell_2^n$\footnote{In fact, in QIT this is usually denoted by $\C^n$.}. Then, a general unit element of this Hilbert space is denoted by $|\psi\rangle$, while notation $\langle\psi|$ is used to denote the same element when it is realized in the dual space. In this way, the standard inner product  $\langle\psi |\varphi\rangle$ gives us the duality action. Also, we can then express the rank-one projection defined by the state $|\psi\rangle$ by $|\psi\rangle\langle\psi|:\ell_2^n\rightarrow \ell_2^n$, so that $|\psi\rangle\langle\psi|(|\varphi\rangle)=|\psi\rangle\langle\psi|\varphi\rangle=\langle\psi|\varphi\rangle|\psi\rangle$. This operators are called \emph{pure states} (and sometimes denoted by the vector $|\psi\rangle$ itself), while general \emph{states} (or \emph{density operators}) are positive operators $\rho:\ell_2^n\rightarrow \ell_2^n$ of trace one. It is also interesting to mention that the elements of the canonical basis of $\ell_2^n$ are usually denoted by $|1\rangle,\cdots ,|n\rangle$. Finally, we will mention that the tensor product symbol is usually omitted. More precisely, if we have two states $|\psi\rangle, |\varphi\rangle\in \ell_2^n$, we will write  $|\psi\rangle|\varphi\rangle\in \ell_2^n\otimes \ell_2^n$ (rather than $|\psi\rangle\otimes |\varphi\rangle$). Moreover, the canonical basis of $\ell_2^n\otimes \ell_2^n$ is usually expressed by using an even more compressed notation, $(|ij\rangle)_{i,j=1}^n$, where $|ij\rangle=|i\rangle|j\rangle=e_i\otimes e_j$ for every $i,j$. In this way, a general norm-one element $|\psi\rangle\in \ell_2^n\otimes_2 \ell_2^n$ can be written as $|\psi\rangle=\sum_{i,j=1}^na_{i,j}|ij\rangle$, where the $a_{i,j}$ are complex coefficients verifying $\sum_{i,j=1}^n |a_{i,j} |^2=1$\footnote{Actually, we can use its Hilbert-Schmidt decomposition to write it as $|\psi\rangle=\sum_{i=1}^n\lambda_i|f_i\rangle|g_i\rangle$ for certain orthonormal systems $(|f_i\rangle)_i$ and $(|g_i\rangle)_i$.} . A particularly interesting example for us is given by the $n$-dimensional maximally entangled state, already introduced in Theorem \ref{Theorem BRSW I}, $\rho:=|\psi_n\rangle\langle\psi_n|$, where $$|\psi_n\rangle:= \frac{1}{\sqrt{n}}\sum_{i=1}^n|ii \rangle\in \ell_2^n\otimes_2\ell_2^n.$$
\subsection{An upper bound for the largest Bell violation of a quantum state}
In order to study the value $LV_\rho$ we will start with an
alternative (somehow dual) statement of \cite[Proposition 2]{JPPVW}
(see also \cite[Theorem 3]{Loubenets} for a related result).
For the sake of completeness we will present a very simple new proof of this
result avoiding operator space terminology. Before, we need to recall
the definition of the projective and injective tensor norms, already
used in several contexts of quantum information theory (see for
instance \cite{HaMo}, \cite{Rudolph}, \cite{Szarek}).

Given a finite dimensional normed space $X$, we denote by $B_X=\big\{x\in X:\|x\|\leq 1\big\}$ its (closed) unit ball. Also, we consider its dual space, $X^*=\Big\{x^*:X\rightarrow \C: x^*   \text{   }\text{ is linear}\Big\}$, with the norm $\|x^*\|= \sup_{x\in B_X}|x^*(x)|$.
If $X$, $Y$ are finite dimensional normed spaces, we will denote the algebraic tensor product by $X\otimes Y$. Then, for a given $u\in X\otimes Y$ we define its \emph{projective tensor norm} as
\begin{align*}
\pi(u)=\inf\Big\{\sum_{i=1}^N\|x_i\|\|y_i\|:N\in \N, u=\sum_{i=1}^Nx_i\otimes y_i\Big\}.
\end{align*}We will denote $X\otimes_\pi Y$ the space $X\otimes Y$ endowed with the projective tensor norm. It is very well known that $\ell_2^n\otimes_\pi\ell_2^n=S_1^n$, where $S_1^n$ is the space of trace class operators from $\ell_2^n$ to $\ell_2^n$. On the other hand, for any $v=\sum_{i=1}^Nx_i\otimes y_i\in X\otimes Y$ we define its \emph{injective tensor norm} as
\begin{align*}
\epsilon(v)=\sup\Big\{\Big|\sum_{i=1}^Nx^*(x_i)y^*(y_i)\Big|:x^*\in B_{X^*}, y^*\in B_{Y^*}\Big\}.
\end{align*}We will denote $X\otimes_\epsilon Y$ the space $X\otimes Y$ endowed with the injective tensor norm.
One can check that the projective and injective tensor norms are dual of each other. Specifically, for any pair of finite dimensional normed spaces $X$, $Y$ we have
\begin{align}\label{duality pi-epsilon}
(X\otimes_\pi Y)^*=X^*\otimes_\epsilon Y^* \text{    }\text{  isometrically}.
\end{align}
In particular, we recover the duality relation $$(S_1^n)^*=(\ell_2^n\otimes_\pi\ell_2^n)^*=\ell_2^n\otimes_\epsilon\ell_2^n=M_n,$$ where $M_n$ denotes the space of maps from $\ell_2^n$ to $\ell_2^n$ with the operator norm.

Finally, we will mention that both tensor norms, projective and injective, can be defined on the tensor product of $N$ spaces exactly in the same way. One can see that Equation (\ref{duality pi-epsilon}) still holds in this general context and, furthermore, both norms are commutative and associative (respect to the spaces in the tensor products).
\begin{proof}[Proof of Theorem \ref{Theorem I}]
Let us consider a quantum strategy constructed with the state $\rho$: $$Q=Q(a,b|x,y)=tr(E_x^a\otimes F_y^b\rho)$$ for every $x,y,a,b$; where $\{E_x^a\}_{x,a}$  and $\{F_y^b\}_{y,b}$ denote POVMs. We do not specify the number of inputs nor outputs because the result will not depend on that. Then, for every $M$ we have
\begin{align*}
\big|\langle M, Q\rangle\big|=\Big|tr\Big(\sum_{x,y;a,b}M_{x,y}^{a,b}E_x^a\otimes F_y^b\rho\Big)\Big|\leq \|\rho\|_{S_1^n\otimes_\pi S_1^n}\Big\|\sum_{x,y;a,b}M_{x,y}^{a,b}E_x^a\otimes F_y^b\Big\|_{M_n\otimes_\epsilon M_n}.
\end{align*}Here, we have used that $(M_n\otimes_\epsilon M_n)^*=S_1^n\otimes _\pi S_1^n$ and the dual action is given by the trace. In order to obtain our statement it suffices to show that
\begin{align}\label{bound 4}
\Big\|\sum_{x,y;a,b}M_{x,y}^{a,b}E_x^a\otimes F_y^b\Big\|_{M_n\otimes_\epsilon M_n}\leq 4\omega(M).
\end{align} To this end, we recall that given en element $\delta$ in the unit ball of $S_1^n$, we can write $\delta=\delta_1+i\delta_2$ where $\delta_1$ and $\delta_2$ are self-adjoint elements in the unit ball of $S_1^n$. Hence, if we denote by $A=(M_n^{sa}, \|\cdot\|_1)$ the space of self adjoint operators with the trace norm, we have
\begin{align*}
\Big\|\sum_{x,y;a,b}M_{x,y}^{a,b}E_x^a\otimes F_y^b\Big\|_{M_n\otimes_\epsilon M_n}\leq 4\sup \Big\{\Big|\sum_{x,y;a,b}M_{x,y}^{a,b}tr(E_x^a\rho_1)tr(F_y^b\rho_2)\Big|:\rho_1,\rho_2\in B_A\Big\}.
\end{align*}Then, we obtain (\ref{bound 4}) by noting that $B_A=conv(S_n\bigcup -S_n)$, where $S_n$ denotes the set of states in $M_n$ and the fact that $\big(tr(E_x^a\delta_1)tr(F_y^b\delta_2)\big)_{x,y;ab}$ is a classical probability distribution for $\delta_1, \delta_2\in S_n$.

In order to see the second assertion in the statement note that, by convexity, it suffices to show it for pure states $\rho=|\psi\rangle\langle \psi|$. On the other hand, we know that  $S_1^n\otimes_\pi S_1^n=\ell_2^n\otimes_\pi\ell_2^n \otimes_\pi\ell_2^n\otimes_\pi\ell_2^n$. Therefore, using that the projective tensor norm does not change if we apply a unitary on each space in the tensor product, one can even assume that our state is diagonal $|\psi\rangle=\sum_{i=1}^n\alpha_i|ii\rangle$ and it is defined with positive coefficients. Furthermore, using the commutativity property of the projective tensor norm we have
\begin{align*}
\big\||\psi\rangle\langle\psi|\big\|_{S_1^n\otimes_\pi
S_1^n}=\Big\|\sum_{i,j=1}^n\alpha_i\alpha_j|ijij\rangle\Big\|_{\bigotimes_{\pi, i=1}^4\ell_2^n}=
\Big\|\sum_{i=1}^n\alpha_i|ii\rangle\otimes
\sum_{j=1}^n\alpha_j|jj\rangle\rangle\Big\|_{\bigotimes_{\pi, i=1}^4\ell_2^n}\\=\Big\|\sum_{i=1}^n\alpha_i|ii\rangle\Big\|^2_{\ell_2^n\otimes_\pi\ell_2^n}=\Big(\sum_{i=1}^n\alpha_i\Big)^2:=\||\psi\rangle\|_1^2.
\end{align*} Since $\sum_{i=1}^n\alpha_i^2=1$, the statement follows from the inequality $\sum_{i=1}^n\alpha_i\leq \sqrt{n}\big(\sum_{i=1}^n\alpha_i^2\big)^{\frac{1}{2}}$.
\end{proof}
In this paper we will restrict to pure states. As it was explained in the previous proof, given a diagonal unit element with positive coefficients\footnote{To compute Bell violations we can always assume that our state is of this form. Indeed, this can be done by composing the corresponding POVMs $(E_x^a)_{x,a}$ and  $(F_y^b)_{y,b}$ with certain unitaries.} $|\varphi\rangle=\sum_{i=1}^n\alpha_i|ii\rangle\in \ell_2^n\otimes_2 \ell_2^n$ and denoting $\rho=|\varphi\rangle\langle\varphi|\in S_1^n\otimes S_1^n$, we have that
\begin{align*}
\|\rho\|_{S_1^n\otimes_\pi S_1^n}=\||\varphi\rangle\|_1^2=\big(\sum_{i=1}^n \alpha_i\big)^2.
\end{align*}Theorem \ref{Theorem I} implies that for every pure state $|\varphi\rangle$ we have
\begin{align*}
LV_{|\varphi\rangle}\leq K\||\varphi\rangle\|_1^2.
\end{align*}
\subsection{Lower bounds for every pure state}\label{lower bounds for every state}
In the remarkable paper \cite{BRSW} the authors showed that the upper bound $O(n)$ given in Theorem \ref{Theorem I} is very tight. Before going on that, we will explain something about \emph{two-prover one-round games (2P1R)-games}. These are particularly interesting Bell inequalities of the form $M_{x,y}^{a,b}=\pi(x,y)V(a,b|x,y)$ for every $x,y=1,\cdots, N$, $a,b=1,\cdots ,K$; where $\pi:[N]\times [N]\rightarrow [0,1]$ is a probability distribution and $V:[K]\times[K]\times[N]\times[N]\rightarrow \{0,1\}$ is a boolean function, usually called \emph{predicate function}. In particular, $2P1R$-games have positive coefficients. These kinds of Bell inequalities are very relevant in computer science because many important problems in complexity theory can be stated in terms of these games. We will keep notation $G=(G_{x,y}^{a,b})_{x,y;a,b=1}^{N,K}$ for these kinds of Bell inequalities. Then, Theorem \ref{Theorem BRSW I} can be stated in the following more precise way.
\begin{theorem}[\cite{BRSW}]\label{Theorem BRSW II}
Let $n$ be a natural number. There exists a game $G_{KV}$ such that
\begin{align*}
LV_{|\psi_n\rangle}(G_{KV})\geq C \frac{n}{(\ln n)^2},
\end{align*}where $|\psi_n\rangle:= \frac{1}{\sqrt{n}}\sum_{i=1}^n|ii\rangle$ is the maximally entangled state in dimension $n$. Here, $C$ is a universal constant which does not depend on the dimension.
\end{theorem}
The game $G_{KV}$ is usually called \emph{Khot-Visnoi game} (or KV game) because it was first defined by Khot and Visnoi to show a large integrality gap for a SDP relaxation of certain complexity problems (see \cite{KhVi} for details). Since the KV game will play an important role in this work we will give a brief description of it (see \cite{BRSW} for a much more complete explanation). For any $n=2^l$ with $l\in\N$ and every $\eta\in [0,\frac{1}{2}]$ we consider the group $\{0,1\}^n$ and the Hadamard subgroup $H$. Then, we consider the quotient group $G=\{0,1\}^n/H$ which is formed by $\frac{2^n}{n}$ cosets $[x]$ each with $n$ elements. The questions of the games $(x,y)$ are associated to the cosets whereas the answers $a$ and $b$ are indexed in $[n]$. The game works as follows: The referee chooses a uniformly random coset $[x]$ and one element $z\in \{0,1\}^n$ according to the probability distribution $pr(z(i)=1)=\eta$, $pr((z(i)=0)=1-\eta$ independently of $i$. Then, the referee asks question $[x]$ to Alice and question $[x\oplus z]$ to Bob. Alice and Bob must answer one element of their corresponding cosets and they win the game if and only if $a\oplus b=z$. Given a probability distribution $P=\big(P([x],[y]|a,b)\big)_{[x],[y]=1;a,b=1}^{\frac{2^n}{n},n}$ it is easy to see that
\begin{align*}
\langle G_{KV},P\rangle=\mathbb{E}_z\frac{n}{2^n}\sum_{[x]}\sum_{a\in[x]}P\big(a,a\oplus z|[x],[x+z]\big).
\end{align*}Now, as a consequence of a clever use of hypercontractive inequality one can see that $\omega(G_{KV})\leq n^{-\frac{\eta}{1-\eta}}$ (see \cite[Theorem 7]{BRSW}). Furthermore, one can define, for any $a\in\{0,1\}^n$, the vector $|u_a\rangle\in \C^n$ by $u_a(i)=\frac{(-1)^{a(i)}}{\sqrt{n}}$ for every $i=1,\cdots, n$. It is easy from the properties of the Hadamard group that $\big(P_a=|u_a\rangle \langle u_a|\big)_{a\in [x]}$ defines a von Neumann measurement\footnote{von Neumann measurements are particular examples of POVMs, where now the operators are orthogonal projections summing up to the identity.} (vNm) for every $[x]$. These measurements will define Alice and Bob's quantum strategies.

A careful study of the KV game shows that for every pure state $|\varphi\rangle$ in dimension $n$ we have
\begin{align}\label{Estimate extension of BRSW II}
LV_{|\varphi\rangle}(G_{KV})\geq C\Big(1+ 4\frac{\||\varphi\rangle\|_1^2-1}{(\ln n)^2}\Big),
\end{align}where $C$ is a universal constant (which can be taken $C=e^{-4}$). This estimate gives us Equation (\ref{Estimate extension BRSW}) when we think of pure states with a large projective norm.
In order to obtain (\ref{Estimate extension of BRSW II}), recall that we can assume that our state is
diagonal with non negative coefficients
$|\varphi\rangle=\sum_{i=1}^n\alpha_i|ii\rangle$. Therefore,
considering the same vNms as above (with respect to the basis
$(|i\rangle)_{i=1}^n$) one can check that the quantum winning
probability is greater than or equal to
\begin{align}\label{KV quantum}
\mathbb{E}_z\Big[\frac{1}{n^2}\frac{n}{2^n}\sum_{[x]}\sum_{a\in [x]}\sum_{i,j=1}^n\alpha_i\alpha_j (-1)^{a(i)}(-1)^{a(j)}(-1)^{a(i)+z(i)}(-1)^{a(j)+z(j)}\Big]\\ \nonumber
=\frac{1}{n}\mathbb{E}_z\Big[\sum_{i,j=1}^n\alpha_i\alpha_j (-1)^{z(i)}(-1)^{z(j)}\Big]=\frac{1}{n}\sum_{i=1}^n\alpha_i^2+\frac{1}{n}\sum_{i\neq j}\alpha_i\alpha_j\mathbb{E}_z \Big[(-1)^{z(i)+z(j)}\Big]\\ \nonumber=\frac{1}{n}\sum_{i=1}^n\alpha_i^2+\frac{1}{n}\sum_{i\neq j}\alpha_i\alpha_j(1-2\eta)^2=\frac{1}{n}+\frac{1}{n}\Big(\big(\sum_{i=1}^n\alpha_i\big)^2-1\Big)(1-2\eta)^2 ,
\end{align}where we have used that $\mathbb{E}_z \Big[(-1)^{z(i)+z(j)}\Big]=(1-2\eta)^2$ is independent of $i,j$ with $i\neq j$.

On the other hand, as we have said before the classical value of
$G_{KV}$ is upper bounded by $n^{-\frac{\eta}{1-\eta}}$. If we
consider $\eta=\frac{1}{2}-\frac{1}{\ln n}\in [0,\frac{1}{2}]$, as
in \cite{BRSW}, we have $n^{-\frac{\eta}{1-\eta}}\leq C \frac{1}{n}$
and the last term in Expression (\ref{KV quantum}) becomes
$\frac{1}{n}+\frac{1}{n}\Big(\big(\sum_{i=1}^n\alpha_i\big)^2-1\Big)(\frac{2}{\ln
n})^2$. Thus, we obtain Equation (\ref{Estimate extension of BRSW II}). 
\begin{remark}\label{general n}
Actually, the KV game is defined for $n=2^l$ with $l$ any natural number. However, an easy modification of the game allows us to state Equation (\ref{Estimate extension of BRSW II}) (so Equation (\ref{Estimate extension BRSW}) too) for a general $n$ with a slight different constant. Indeed, for a given state $|\varphi\rangle$ in dimension $n$ we define $l_0=\max\{l:2^l\leq n\}$. Then, we can consider the KV game in dimension $m=2^{l_0}$ and artificially add an extra $m+1$ output for Alice and Bob so that the predicate function of the game is always zero for these new values. Then, the only difference in the classical value of the game is that we must optimize over all families of non negative numbers $\big(P(a|x)\big)_{x,a}$, $\big(Q(b|y)\big)_{y,b}$ such that $\sum_aP(a|x)\leq 1$ for every $x$ and $\sum_bQ(b|y)\leq 1$ for ever $y$. However, since all coefficients of the game are positive it is trivial to deduce that the optimum families will verify equality in the previous expressions. Therefore, the classical value of the new game is exactly the classical value of the KV game in dimension $m$. On the other hand, for every $a\in\{0,1\}^m$ we can define the vector $|u_a\rangle\in \C^n$, by $u(i)=\frac{(-1)^{a(i)}}{\sqrt{m}}$ if $1\leq i\leq m$ and $u(i)=0$ otherwise. Then, the same calculation as above shows that if we consider the quantum probability distribution $Q$ constructed with the state $|\varphi\rangle$ and the von Neumann measurements defined as$$\Big\{\big(P_{[x]}^a=|u_a\rangle\langle u_a|\big)_{a\in [x]}, P_{[x]}^{m+1}=\uno_{M_n}-\sum_{a\in [x]}|u_a\rangle\langle u_a|\Big\}   \text{      } \text{      }\text{(and similar for Bob),}$$we obtain that $$\langle G_{KV},Q\rangle=\frac{1}{m}\Big(\sum_{i=1}^m\alpha_i^2\Big)\Big(1-\big(1-2\eta\big)^2\Big)+\frac{1}{m}\Big(\sum_{i=1}^m\alpha_i\Big)^2\big(1-2\eta\big)^2.$$ Then, considering $\eta=\frac{1}{2}-\frac{1}{\ln m}$ and using that $n\geq m=2^{l_0}\geq \frac{n}{2}$ we recover the same estimates as in (\ref{Estimate extension of BRSW II}) with a slight modification in the constant.
\end{remark}
Since we are looking for a good measure of
$LV_{|\varphi\rangle}$ for a general pure state $|\varphi\rangle$,
we must be careful about giving lower bounds depending on the rank
(or dimension) of the state. Indeed, in many cases this can distort
the essence of a state. With the computations above and the same
ideas as in Remark \ref{general n} it is easy to see that one can
give the following better lower bound for the largest Bell violation
of a pure state $|\varphi\rangle=\sum_{i=1}^n\alpha_i|ii\rangle$,
\begin{align*}
LV_{|\varphi\rangle}\geq C\sup_{k=1,\cdots, n}\Big(\frac{\sum_{i=1}^k\alpha_i}{\ln k}\Big)^2,
\end{align*}where $C$ is a universal constant.

Note that for the maximally entangled state we obtain $LV_{|\psi_n\rangle}\geq C\frac{n}{(\ln n)^2}$ as it is stated in Theorem \ref{Theorem BRSW I} (resp. Theorem \ref{Theorem BRSW II}). The previous study shows the projective tensor norm as a good candidate to measure the largest Bell violation attainable by a (pure) state. In fact, regarding the previous results, one can wonder whether we can improve Buhrman's et al result to obtain
\begin{align*}
LV_{|\psi_n\rangle}\geq C n.
\end{align*}
Our Theorem \ref{Main I} shows that this is not possible. That is, we
cannot completely remove the logarithmic factor in Theorem \ref{Theorem BRSW I}.  Theorem \ref{Main I} (joint with Theorem \ref{Theorem BRSW I}) clarifies the asymptotic behavior of the largest Bell violation of the maximally entangled state up to the order of the logarithmic factor:
\begin{align*}
C\frac{n}{(\ln n)^2}\leq LV_{|\psi_n\rangle}\leq D\frac{n}{\sqrt{\ln n}}.
\end{align*}
We must mention that when we restrict to the easier case of von
Neumann measurements (vNms) rather than general POVMs one can
improve the upper bound in Theorem \ref{Main I} to obtain
$O(\frac{n}{\ln n})$. Indeed, following Werner's construction
(\cite{Werner}), in \cite{APBTA} the authors showed that for certain
$p\geq K\frac{\ln n}{n}$ (with $K\geq 0.8$) the state
$$\xi_p=p|\psi_n\rangle\langle\psi_n|+ (1-p)\frac{\uno}{n^2}$$ is
\emph{vNm-local}; that is, one can construct a local hidden variable
model to describe any quantum probability distribution
$\big(tr(P_x^a\otimes Q_y^b\xi_p)\big)_{x,y}^{a,b}$ constructed with
vNms $\{P_x^a\}$, $\{Q_y^b\}$. Here, we denote by $\frac{\uno}{n^2}$
the maximally mixed state. Since $\frac{\uno}{n^2}$ is a separable
state, one immediately deduces that
\begin{align*}
LV_{|\psi_n\rangle}^{vN}\leq K'\frac{n}{\ln n},
\end{align*}where $LV_{|\psi\rangle}^{vN}$ denotes the measure $LV_{|\psi\rangle}$ restricted to quantum probability distributions constructed by applying vNms on the state $|\psi\rangle$, and $K'\leq\frac{5}{2}$. We must mention, however, that restricting to vNms, though very natural from a physical point of view, simplifies very much the geometry of the problem. Actually, the best estimate in \cite{APBTA} for $p$ verifying that
$\xi_p$ is \emph{local} (with general POVMs) is
$\Omega(\frac{1}{n})$, which leads to an estimate
$LV_{|\psi_n\rangle}\leq D n$. It is also worth mentioning that the
KV game can be used to improve the upper bound estimates in
\cite{APBTA}. Indeed, since the quantum strategy used in Theorem
\ref{Theorem BRSW I} is constructed with vNms acting on the maximally
entangled state in dimension $n$, we immediately conclude that
$\frac{(\ln n)^2}{n}$ is an upper bound for the value $p_L^\phi$
considered in \cite{APBTA}.
\section{A relaxation of the problem}\label{Section: Relaxation}
Let us consider the following SDP relaxation for the quantum value of a $2P1R$-game $G$ with $N$
questions and $K$ answers, which optimizes over families of real
vectors $\{u_x^a\}_{x,a=1}^{N,K}$, $\{v_y^b\}_{y,b=1}^{N,K}$:
\begin{align}\label{SDP}
\begin{array}{lll}
                 SDP(G):= & \text{\textbf{Maximize:}} & \Big|\sum_{x,y;a,b=1}^{N,K}G_{x,y}^{a,b}\langle u_x^a, v_y^b\rangle\Big|\\
                 &  \text{\textbf{Subject to:}}  & \forall x, \forall a\neq a',\langle u_x^a,u_x^{a'}\rangle=0 \text{    }\text{  and   }\text{    }\forall y, \forall b\neq b',\langle v_y^b,v_y^{b'}\rangle=0, \\
                 &            & \forall x, \|\sum_au_x^a\|= 1 \text{    }\text{  and   }\text{    }\forall y, \|\sum_bv_y^b\|= 1.
                           \end{array}\end{align}
The orthogonality restriction in (\ref{SDP}) comes from the fact that any quantum
probability distribution $Q\in \mathcal Q$ can be written by using von Neumann measurements.  However, this process involves an increase
in the dimension of the Hilbert spaces. Therefore, this constrain \emph{is not natural} when one is interested in studying the dimension of the considered quantum states; as we are in this work.  Furthermore, since we are interested here in fixing the dimension of our quantum states $\rho$, we would like to truncate the previous SDP relaxation by
requiring the families of vectors $\{u_x^a\}_{x,a=1}^{N,K}$, $\{v_y^b\}_{y,b=1}^{N,K}$ to have a fixed dimension $n$. However, this is not possible in general if we want to preserve 
the orthogonality restrictions, since this restriction implies that $K$ must be smaller than or equal to $n$ (while we are typically interested in the opposite case). In order to save this problem we will consider the following optimization problem, which optimizes over families of real vectors
$\{u_x^a\}_{x,a=1}^{N,K}$, $\{v_y^b\}_{y,b=1}^{N,K}$:
\begin{align}\label{OPinfty}
\begin{array}{lll}
                 \omega_{OP_\infty}(G):= & \text{\textbf{Maximize:}} & \Big|\sum_{x,y,a,b=1}^{N,K}G_{x,y}^{a,b}\langle u_x^a, v_y^b\rangle\Big|\\
                    &  \text{\textbf{Subject to:}}  & \forall x, \sup_{\alpha_x^a=\pm 1}\big\|\sum_{a=1}^K\alpha_x^a u_x^a\big\|\leq 1, \\
                 &            & \forall y, \sup_{\alpha_y^b=\pm 1}\big\|\sum_{b=1}^K\alpha_y^b v_y^b\big\|\leq 1.
                           \end{array}\end{align}
Then, it is very easy to see that $\omega_{OP_\infty}(G)$ is a
relaxation for the problem of computing the quantum value of a
$2P1R$-game $G$ and it verifies that $SDP(G)\leq
\omega_{OP_\infty}(G)$ for every $G$. The value
$\omega_{OP_\infty}(G)$ is a natural generalization of $SDP(G)$
which removes the orthogonality restriction and so, it admits
restrictions in the dimension of the vectors. We will call
$\omega_{OP_n}(G)$ the value of the previous optimization problem
with the extra restriction: $u_x^a, v_y^b\in \R^n$ for every
$x,y,a,b$.
Then, we have
\begin{theorem}\label{Main III}
\begin{align*}
\omega_{OP_n}(G)\leq D\frac{n}{\sqrt{\ln n}}\omega(G)
\end{align*}for every $2P1R$-game $G$, where $D$ is a universal constant.
\end{theorem}
We think that Theorem \ref{Main III} and Theorem \ref{Main II}  can be of independent interest for computer scientists.
$\omega_{OP_n}$ is the natural generalization of $SDP$ when we want
to impose ``low dimensional solutions''(where orthogonality
restrictions no longer make sense since we will have $K> n$). As far
as we know the question of rounding low-dimensional solutions of
these kinds of optimization problems has not received much
attention. Some interesting papers in this direction are
\cite{AvZw}, \cite{BFV}, \cite{BFVII}.

Since in this paper we want to work in the general context of Bell
inequalities (rather than restricting to the specific case of
$2P1R$-games) we have to consider a modification of the definition
of $\omega_{OP_\infty}$ (resp. $\omega_{OP_n}$). Indeed, the
non-signaling condition verified by the classical and quantum
probability distributions plays an important role in this case and
one has to impose an extra restriction in the definition of $\omega_{OP_n}$ to avoid trivial cases where
$\omega(M)=0$ and $\omega_{OP_n}(M)> 0$, which makes not possible
any result like Theorem \ref{Main II} (see \cite[Section 5]{JP} for
a complete study on the geometry of the problem). Then, for a given
Bell inequality $M\in \mathcal M^{N,K}$, we consider the optimization problem presented in the introduction, which optimizes over families of real $n$
dimensional vectors $\{u_x^a\}_{x,a=1}^{N,K}$,
$\{v_y^b\}_{y,b=1}^{N,K}$, $z$:
\begin{align*}
\begin{array}{lll}
                 \overline{\omega}_{OP_n}(M):= & \text{\textbf{Maximize:}} & \Big|\sum_{x,y,a,b=1}^{N,K}M_{x,y}^{a,b}\langle u_x^a, v_y^b\rangle\Big|\\
                    &  \text{\textbf{Subject to:}}   & \forall x,y, \sum_{a=1}^K u_x^a=\sum_{b=1}^K v_y^b=z \text{   }  (*),\\
                     &            & \forall x, \sup_{\alpha_x^a=\pm 1}\big\|\sum_{a=1}^K\alpha_x^a u_x^a\big\|\leq 1, \\
                 &            & \forall y, \sup_{\alpha_y^b=\pm 1}\big\|\sum_{b=1}^K\alpha_y^b v_y^b\big\|\leq 1.
                           \end{array}\end{align*}
We re-state here our main Theorem \ref{Main II}.
\begin{theorem*}
For all natural numbers $n$, $N$, $K$ and every $M\in \mathcal M^{N,K}$ we have
\begin{align*}
\overline{\omega}_{OP_n}(M)\leq D\frac{n}{\sqrt{\ln n}}\omega(M),
\end{align*}where $D$ is a universal constant.
\end{theorem*}
The proofs of Theorem \ref{Main II} and Theorem \ref{Main III} are the same, but in the first case we have the extra difficulty of restricting to a certain ``affine subspace'' described by condition (*). In particular, Theorem \ref{Main III} can be obtained by following exactly the same proof as the one we will give for Theorem \ref{Main II} with obvious modifications. We will postpone the proof of Theorem \ref{Main II} to Section \ref{Section: proofs} and we will finish this section by showing how to obtain Theorem \ref{Main I} from Theorem \ref{Main II}.
\begin{proof}[Proof of Theorem \ref{Main I}]
Let us consider a Bell inequality $M\in \mathcal M^{N,K}$ such that $\omega(M)\leq 1$. We must show that
\begin{align*}
\Big|\sum_{x,y;a,b=1}^{N,K}M_{x,y}^{a,b}Q_{x,y}^{a,b}\Big|\leq D\frac{n}{\sqrt{\ln n}}
\end{align*}for every $Q\in \mathcal Q_{|\psi_n\rangle}^{N,K}$. By definition $Q_{x,y}^{a,b}=tr\big(E_x^a\otimes F_y^b|\psi_n\rangle\langle\psi_n|\big)$ for every $x,y,a,b$, where $\{E_x^a\}_{x,a}$ and  $\{F_y^b\}_{y,b}$ are POVMs in dimension $n$ and $|\psi_n\rangle=\frac{1}{\sqrt{n}}\sum_{i=1}^n|ii\rangle$. Then, we can write
\begin{align*}
tr\big(E_x^a\otimes F_y^b|\psi_n\rangle\langle\psi_n|\big)=\frac{1}{n}\sum_{i,j=1}^nE_x^a(i,j)F_y^b(i,j)=\frac{1}{n}\sum_{i,j=1}^n(E_x^a)^t(j,i)F_y^b(i,j)
=\frac{1}{n}tr\big((E_x^a)^tF_y^b\big),
\end{align*}where $^t$ denotes the transpose.
Therefore,
\begin{align}\label{basic}
\Big|\sum_{x,y;a,b}M_{x,y}^{a,b}Q_{x,y}^{a,b}\Big|= \Big |\frac{1}{n}\sum_{i=1}^n\sum_{x,y;a,b}M_{x,y}^{a,b}\big((E_x^a)^tF_y^b\big)(i,i)\Big| \nonumber=\Big |\frac{1}{n}\sum_{i=1}^n\sum_{x,y;a,b}M_{x,y}^{a,b}\langle u_x^{a,i},v_y^{b,i}\rangle\Big|\\\leq \sup_{i=1,\cdots,n}\Big |\sum_{x,y;a,b}M_{x,y}^{a,b}\langle u_x^{a,i},v_y^{b,i}\rangle\Big|,
\end{align}where $|u_x^{a,i}\rangle=\overline{E_x^a}|i\rangle$ for every $x,a,i$ and $|v_y^{b,j}\rangle=F_y^b|j\rangle$ for every $y,b,j$.

Note that for a fixed $i=1,\cdots, n$, we trivially have
\begin{align}\label{prop1}
\sum_{a=1}^Ku_x^{a,i}=\sum_{b=1}^Kv_y^{b,i}=|i\rangle
\end{align}for every $x,y$. Furthermore, for every $x$ and every $(\alpha_a)_{a=1}^K\in\{-1,1\}^K$, we have that
\begin{align}\label{prop2}
\Big\|\sum_{a=1}^K\alpha_au_x^{a,i}\Big\|=\Big\|\sum_{a=1}^K\alpha_a\overline{E_x^a}|i\rangle \Big\|\leq \Big\|\sum_{a=1}^K\alpha_a\overline{E_x^a}\Big\|_{M_n}\leq 1;
\end{align}and analogously for the $v_y^{b,i}$'s, where the last inequality follows from the fact that $\{E_x^a\}_{a}$ is a POVM.
Therefore, for every $i=1,\cdots,n$ the families of $n$-dimensional (possibly complex) vectors $\{u_x^{a,i}\}_{x,a}$ and $\{v_y^{b,i}\}_{y,b}$ verify the conditions in (\ref{restrictions}). The only thing left to do is to show that these vectors can be assumed to be real. Indeed, if this is true, we can apply Theorem \ref{Main II} to conclude
\begin{align*}
\Big|\sum_{x,y;a,b}M_{x,y}^{a,b}Q_{x,y}^{a,b}\Big|\leq \sup_{i=1,\cdots,n}\Big |\sum_{x,y;a,b}M_{x,y}^{a,b}\langle u_x^{a,i},v_y^{b,i}\rangle\Big|\leq D\frac{n}{\sqrt{\ln n}}.
\end{align*}
We can assume the families $\{u_x^{a,i}\}_{x,a}$ and $\{v_y^{b,i}\}_{y,b}$ to be formed by real vectors by replacing $n$ with $2n$ (which means just a slight modification in the constant $D$). To see this we note that Equation (\ref{basic}) can be read as
\begin{align*}
\Big|\sum_{x,y;a,b}M_{x,y}^{a,b}Q_{x,y}^{a,b}\Big|=\Big|Re\Big(\sum_{x,y;a,b}M_{x,y}^{a,b}Q_{x,y}^{a,b}\Big)\Big|\leq \sup_{i=1,\cdots,n}\Big |\sum_{x,y;a,b}M_{x,y}^{a,b}Re(\langle u_x^{a,i},v_y^{b,i}\rangle)\Big|,
\end{align*}where $Re(z)$ denote the real part of $z$. On the other hand, if we define the vectors $\widetilde{u}_x^{a,i}=Re(u_x^{a,i})\oplus Im(u_x^{a,i})\in \R^{2n}$ and $\widetilde{v}_y^{b,j}=Re(v_y^{b,j})\oplus Im(v_y^{b,j})\in \R^{2n}$, we obtain new real vectors verifying $\langle \widetilde{u}_x^{a,i},\widetilde{v}_y^{b,i}\rangle=Re(\langle u_x^{a,i},v_y^{b,i}\rangle)$ for every $x,a$ and also conditions (\ref{prop1}) and  (\ref{prop2}).
So we have done.
\end{proof}
\section{proof of the main result}\label{Section: proofs}
\subsection{Proof of Theorem \ref{Main II}}\label{proof main}
The proof of Theorem \ref{Main II} will follow the same lines as
\cite[Theorem 18]{JP}. However, we will present here a simpler
approach to the problem avoiding, in particular, the use of
\cite{BCLT} (via \cite[Theorem 19]{JP}). Our proof relies
on Lemma \ref{concentration} proven below. In order to make the
proof of Theorem \ref{Main II} completely understandable for every
reader, we will start by introducing a few definitions and basic
results. In the following we will denote by $\ell_2^n$ the space
$\R^n$ with the Euclidean norm\footnote{Note that we used this
notation so far to denote the complex
$n$-dimensional Hilbert space. Here, we will restrict to real
spaces.} and by
$\ell_\infty^n$ the space $\R^n$ with the $\sup$-norm. We will denote by
$\ell_2$ and $\ell_\infty$ the corresponding infinite dimensional
spaces.
On the other hand, given a linear map $T:X\rightarrow Y$ between two finite dimensional normed spaces, we will denote the norm of $T$ by $$\|T\|:=\sup_{x\in B_X}\|T(x)\|_Y,$$where $B_X$ is the unit ball of $X$. Note that for a linear map $T:\ell_\infty^K\rightarrow \ell_2^n$ defined as $T(|a\rangle)=|u_a\rangle$ for every $a$, where $(|a\rangle)_{a=1}^K$ denotes the standard basis in $\R^K$, we have
\begin{align}\label{norm infty-two}
\|T\|=\sup_{(\alpha_a)_a\in \{-1,1\}^K}\Big\|\sum_{a=1}^K\alpha_a|u_a\rangle\Big\|_{\ell_2^n}=
\sup_{\sum_{i=1}^n|\beta_i|^2=1}\sum_{a=1}^K\Big|\sum_{i=1}^n\beta_i\langle u_a|i\rangle\Big|.
\end{align}
In the particular case where $T:\ell_2^n\rightarrow Y$, we will be also interested in the following norm of $T$,
\begin{align*}
\ell(T):=\mathbb{E}\Big(\Big\|\sum_{i=1}^ng_iT(|i\rangle)\Big\|_Y^2\Big)^\frac{1}{2},
\end{align*}where $(|i\rangle)_{i=1}^n$ denotes the standard basis of $\ell_2^n$ and $(g_i)_{i=1}^n$ is a sequence of independent normalized real random Gaussian variables.
An easy computation shows
\begin{align}\label{l-norm T}
\ell(T)=\mathbb{E}\Big(\big\|T\big(\sum_{i=1}^ng_i|i\rangle\big)\big\|_Y^2\Big)^\frac{1}{2}\leq \|T\|\mathbb{E}\Big(\Big\|\sum_{i=1}^ng_i|i\rangle\Big\|_{\ell_2^n}^2\Big)^\frac{1}{2} \leq \sqrt{n}\|T\|
\end{align}for every linear map $T:\ell_2^n\rightarrow Y$.
According to Kahane- Khinchin inequality (see for instance \cite{Tomczak}, pp 16) we know that
\begin{align}\label{Kahane-Khinchin}
\ell(T)\leq
K_{1,2}\mathbb{E}\Big\|\sum_{i=1}^ng_iT(|i\rangle)\Big\|_Y
\end{align}for every $T:\ell_2^n\rightarrow Y$, where $K_{1,2}$ is a universal constant.

Finally we will introduce a third norm for a given linear map $T:X\rightarrow Y$. We say that $T$ is \emph{2-summing} if there exists a constant $C\geq 0$ such that for every $N\in \N$ and every sequence of elements $x_1,\cdots ,x_N$ in $X$ the following inequality holds:
\begin{align}\label{2-summing}
\Big(\sum_{i=1}^N\|T(x_i)\|_Y^2\Big)^\frac{1}{2}\leq C \sup_{x^*\in B_{X^*}}\Big(\sum_{i=1}^N|x^*(x_i)|^2\Big)^\frac{1}{2}.
\end{align}In this case we define the 2- summing norm of $T$ as $\pi_2(T):=\inf \{C: C \text{  verifies  } (\ref{2-summing})\}.$ A particularly simple case is when $T:\ell_\infty\rightarrow \ell_\infty$ is a diagonal map defined by a sequence $(\lambda_i)_{i=1}^\infty$ (that is, $T(|i\rangle)=\lambda_i|i\rangle$ for every $i$). In this case, one can see that $\pi_2(T)=\|(\lambda_i)_{i=1}^\infty\|_2$. It is also easy to verify from its definition that the 2-summing operators form an operator ideal. In particular, for all linear maps between Banach spaces $T:X\rightarrow Y$, $S:Y\rightarrow Z$ and $Q:Z\rightarrow W$ we have that $\pi_2(Q\circ S\circ T)\leq \|Q\|\pi_2(S)\|T\|$.
\emph{Grothendieck inequality} has been already used in several problems of quantum information theory (see \cite{Pisiersurvey} for a complete survey of the topic). As an immediate consequence of Grothendieck inequality we deduce that for every linear map $T:\ell_\infty\rightarrow \ell_2$ we have
\begin{align}\label{Gro}
\pi_2(T)\leq K_G\|T\|,
\end{align}where $K_G$ is the \emph{Grothendieck constant}, which is known to verify $K_G< 1.78$\footnote{Actually, to state inequality $(\ref{Gro})$ it is enough to invoke the \emph{little Grothendieck theorem} which gives us a constant $\sqrt{\frac{\pi}{2}}$.}. The following inequality will be very helpful in the proof of Lemma \ref{concentration}. Let $a:\ell_2\rightarrow \ell_\infty$ and $b:\ell_\infty\rightarrow \ell_2$ be two linear maps, then \begin{align}\label{trece+gro}
\big|tr(b\circ a)|\leq \pi_2(b)\pi_2(a)\leq K_G\|b\|\pi_2(a).
\end{align}Here, the first inequality is a consequence of \emph{trace duality} (see for instance \cite{DeFl}) and the second one follows from Equation (\ref{Gro}).

The following lemma will be crucial in the proof of Theorem \ref{Main II}.
\begin{lemma}\label{concentration}
For all natural numbers $n,N\in \mathbb{N}$ and all linear maps $S:\ell_2^n\rightarrow \ell_\infty^N$ and $T:\ell_\infty^N\rightarrow \ell_2^n$ we have
\begin{align*}
\big|tr(T\circ S)\big|\leq C\sqrt{\frac{n}{\ln n}}\|T\|\ell(S),
\end{align*}where $C$ is a universal constant.
\end{lemma}
The key point in the proof of Lemma \ref{concentration} is a nice
consequence of the concentration of measure phenomenon given by
Ledoux and Talagrand. It has already been used in the study of
cotype constants in Banach space theory. In particular, we develop
here some ideas from \cite{JuGe}.
\begin{proof}
According to \cite[Theorem 12.10]{LeTa} applied to the Gaussian process $X_t=\sum_{i=1}^ng_i \langle t|S|i\rangle$, $t=1,\cdots, N$, there exists a Gaussian sequence $(Y_k)_{k\geq 1}$ with $\|Y_k\|_2\leq C\frac{\ell(S)}{\sqrt{\ln (k+1)}}$ for every $k\geq 1$ and such that for every $t=1,\cdots, N$ we have $$X_t=\sum_{k\geq 1}\alpha_k(t)Y_k,$$where $\alpha_k(t)\geq 0$, $\sum_{k\geq 1}\alpha_k\leq 1$ and the series converges almost surely in $L_2$. Then, for every $k\geq 1$ we can define $u_k=\sqrt{\ln(k+1)}\sum_{i=1}^n\langle Y_k, g_i\rangle |i\rangle\in \ell_2^n$, $v_k=\sum_{t=1}^N\alpha_k(t)|y\rangle\in \ell_\infty^N$ and the previous properties guarantee that $\|u_k\|\leq C\ell(S)$ and $\|v_k\|\leq 1$ for every $k$. Let us consider now the linear maps $A:\ell_2^n\rightarrow \ell_\infty$, $D:\ell_\infty\rightarrow \ell_\infty$ and $B:\ell_\infty\rightarrow \ell_\infty^N$ define by $A(|i\rangle)=\sum_{k\geq 1}\langle u_k|i\rangle |k\rangle$ for every $i=1,\cdots ,n$; $D(|k\rangle)=\frac{1}{\sqrt{\ln(k+1)}}|k\rangle$ for every $k\geq 1$ and $B(|k\rangle)=v_k$ for every $k\geq 1$ respectively. Cauchy-Schwartz inequality implies that $\|A\|\leq C'\ell(s)$, whereas it is easy to check that $\|D\|\leq 1$ and $\|B\|\leq 1$. Furthermore, the following factorization holds: $$S=B\circ D\circ A.$$Following \cite[Lemma 3.3]{JuGe} we write $D=D_1+D_2$ where $D_1$ is the diagonal operator associated to the sequence $\tilde{D}_1=(\frac{1}{\sqrt{\ln 2}},\cdots , \frac{1}{\sqrt{\ln (n+1)}},0,0, \cdots)$. Then, we have
\begin{align*}
|tr(T\circ S)|=|tr(T\circ B\circ D\circ A)|\leq |tr(T\circ B\circ D_1\circ A)|+ |tr(T\circ B\circ D_2\circ A)|.
\end{align*}Now, according to Equation (\ref{trece+gro}) and the ideal property of 2-summing operators we have
\begin{align*}
\big|tr(T\circ B\circ D_1\circ A)\big|\leq K_G \pi_2(D_1\circ A)\|T\circ B\|\leq K_G \pi_2(D_1)\|A\|\|T\|\|B\|\leq C'' \sqrt{\frac{n}{\ln n}}\ell(S)\|T\|,
\end{align*}where we have used $\pi_2(D_1)=\|\tilde{D}_1\|_2\leq \tilde{C}\sqrt{\frac{n}{\ln n}}$. On the other hand, if we denote $id_n:\ell_2^n\rightarrow \ell_\infty^n$ the identity map, we have
\begin{align*}
\big|tr(T\circ B\circ D_2\circ A)\big|=\big|tr(id_n^{-1}\circ id_n\circ T\circ B\circ D_2\circ A)\big|\leq K_G\pi_2(id_n\circ T\circ B\circ D_2\circ A)\|id_n^{-1}\|\\ \leq K_G\|id_n\|\pi_2(T)\|B\|\|D_2\|\|A\|\|id_n^{-1}\|\leq C'''\sqrt{\frac{n}{\ln n}}\|T\|\ell(S),
\end{align*}where we have used that $\|D_2\|\leq \frac{1}{\sqrt{\ln (n+1)}}$, $\pi_2(T)\leq K_G\|T\|$, $\|id_n\|=1$ and $\|id_n^{-1}\|=\sqrt{n}$.
Therefore, we obtain that
\begin{align*}
|tr(T\circ S)|\leq \tilde{C}\sqrt{\frac{n}{\ln n}}\|T\|\ell(S),
\end{align*}as we wanted.
\end{proof}
\begin{remark}\label{optimality I}
We note that Lemma \ref{concentration} is optimal. Indeed, if we
consider the map $id_n:\ell_2^n\rightarrow \ell_\infty^n$ it is well
known that $\ell(id_n)\leq c\sqrt{\ln n}$ for some universal
constant $c$ and we also have $\|id_n^{-1}\|=\sqrt{n}$. On the other
hand, we trivially have $tr(id_n^{-1}\circ id_n)=n$.
\end{remark}
\begin{lemma}\label{preliminar lemma}
Let $R=(R(x|a))_{x,a=1}^{N,K}$ be a family of real numbers such that $\sum_{a=1}^KR(x|a)=C$ for every $x=1,\cdots, N$, where $C$ is a constant. Let us denote $\Lambda=\sup_{x=1,\cdots ,N}\sum_{a=1}^K|R(x|a)|$. Then, we can write $R=\lambda P_1+ \mu P_2$ such that $P_i\in S(N,K)$ for $i=1,2$ and $|\lambda|+ |\mu|=\Lambda$. Here, we denote $$S(N,K)=\Big\{(P(x|a))_{x,a=1}^{N,K}: P(x|a)\geq 0   \text{    and     }\sum_{a=1,\cdots, K}P(x|a)=1   \text{    for every   } x,a\Big\}.$$
\end{lemma}
\begin{proof}
We can assume the constant $C$ to be positive. For every $x$, we denote $$A_x^+=\{a: R(x|a)> 0\} \text{   }\text{ and  }\text{   }A_x^-=\{a: R(x|a)\leq 0\}.$$ Also, we denote $$M=\sup_x\sum_{a\in A_x^+}R(x|a)\text{   }\text{ and  }\text{   }m=\inf_x\sum_{a\in A_x^-}R(x|a).$$Since the case $m=0$ is trivial we can assume that $m< 0$. The fact that $\sum_aR(x|a)=C$ for every $x$ guarantees that the previous $\sup$ and $\inf$ are attained in the same $x$. In particular note that $M+m=C$ and $M-m=\Lambda$. Therefore, we can write $R=MP_1+ mP_2$, where we define, for each $x$: $P_1(a|x)=\frac{R(a|x)}{M}$ for $a\in \{1,\cdots, K-1\}\cap A_x^+$, $P_1(a|x)=0$ for $a\in \{1,\cdots, K-1\}\cap A_x^-$, $P_1(K|x)=1-\sum_{a=1}^{k-1}P_1(a|x)$ and $P_2(a|x)=\frac{R(a|x)}{m}$ for $a\in \{1,\cdots, K-1\}\cap A_x^-$, $P_2(a|x)=0$ for $a\in \{1,\cdots, K-1\}\cap A_x^+$, $P_2(K|x)=1-\sum_{a=1}^{k-1}P_2(a|x)$.
Since $P_1$ and $P_2$ belong to $S(N,K)$ and $|M|+|m|=\Lambda$ we conclude the proof.
\end{proof}
We are now ready to prove our main result.
\begin{proof}[Proof of Theorem \ref{Main II}]
Let us consider an element $M\in \mathcal M^{N,K}$ such that $\omega(M)\leq 1$ and some families of vectors $\{u_x^a\}_{x,a}$ and $\{v_y^b\}_{y,b}$ verifying conditions (\ref{restrictions}). We must show that
\begin{align*}
\Big|\sum_{x,y,a,b}M_{x,y}^{a,b}\langle u_x^a,v_y^b\rangle\Big|\leq D\frac{n}{\sqrt{\ln n}}
\end{align*}for a certain universal constant $D$.

In order to fit Lemma \ref{concentration} in our context we must ``twist'' our Bell inequality $M$ in the spirit of \cite[Section 5]{JP}. For every fixed $y=1,\cdots, N$, we consider the linear maps $$u_y:\ell_2^n\rightarrow \ell_\infty^{K-1}\text{   }\text{   defined by   }\text{    } u_y(|i\rangle)=\sum_{b=1}^{K-1}\sum_{x,a=1}^{N,K}\Big(M_{x,y}^{a,b}-M_{x,y}^{a,K}\Big)u_x^a(i)|b\rangle \text{   }\text{   for every   }1\leq i\leq n$$and $$u_{N+1}:\ell_2^n\rightarrow \ell_\infty^{K-1}   \text{   }\text{   defined by   }\text{    } u_{N+1}(|i\rangle)=\sum_{y=1}^N\sum_{x,a=1}^{N,K}M_{x,y}^{a,K}u_x^a(i)|1\rangle \text{   }\text{   for every   }1\leq i\leq n.$$On the other hand, we will also consider the linear maps $$v_y:\ell_\infty^{K-1}\rightarrow \ell_2^n\text{   }\text{   defined by   }\text{    } v_y(|b\rangle)=v_y^b \text{   }\text{   for every   }1\leq b\leq K-1,$$and$$v_{N+1}:\ell_\infty^{K-1}\rightarrow \ell_2^n   \text{   }\text{   defined by   }\text{    } v_{N+1}(|b\rangle)=\left\{\begin{array}{ccccc}
0 & \text{ if } & |b\rangle\neq |1\rangle\\
\sum_{b=1}^Kv_1^b & \text{ if } & |b\rangle= |1\rangle.
\end{array}\right.$$
Then, trivial computations show that
\begin{align*}
\sum_{x,y;a,b=1}^{N,K}M_{x,y}^{a,b}\langle u_x^a,v_y^b\rangle= \sum_{y=1}^{N+1}tr(v_y\circ u_y).
\end{align*}
Now, according to Equation (\ref{norm infty-two}) and conditions (\ref{restrictions}) we have that $\|v_y\|\leq 1$ for every $y=1,\cdots, N+1$.
Therefore, according to Lemma \ref{concentration} we have
\begin{align*}
\Big|\sum_{x,y;a,b=1}^{N,K}M_{x,y}^{a,b}\langle u_x^a,v_y^b\rangle\Big|\leq \sum_{y=1}^{N+1}\Big|tr(v_y\circ u_y)\Big|\leq C \sqrt{\frac{n}{\ln n}}\sum_{y=1}^{N+1}\ell(u_y).
\end{align*}Our statement will follow then from the estimate
\begin{align}\label{final estimate}
\sum_{y=1}^{N+1}\ell(u_y)\leq 3K_{1,2}\sqrt{n}.
\end{align}
First, according to Equation (\ref{Kahane-Khinchin}) we have
\begin{align*}
\sum_{y=1}^{N}\ell(u_y)\leq
K_{1,2}\sum_{y=1}^{N}\mathbb{E}\Big\|\sum_{i=1}^ng_iu_y(|i\rangle)\Big\|_{\ell_\infty^{K-1}}
=K_{1,2}\sum_{y=1}^{N}\mathbb{E}\sup_{b=1,\cdots,K-1}\Big|\sum_{x,a=1}^{N,K}\Big(M_{x,y}^{a,b}-M_{x,y}^{a,K}\Big)\sum_{i=1}^ng_iu_x^a(i)\Big|.
\end{align*}
Now, let us denote, for every $y=1,\cdots, N$, $b_y\in \{1,\cdots, K-1\}$ the elements where the previous $\sup$ is attained and $\alpha_y=\sign\Big(\sum_{x,a=1}^{N,K}\Big(M_{x,y}^{a,b_y}-M_{x,y}^{a,K}\Big)\sum_{i=1}^ng_iu_x^a(i)\Big)$. Then, defining the element $(R(b|y))_{y,b=1}^{N,K}$ by $R(b|y)=\alpha_y$ if $b=b_y$, $R(K|y)=-\alpha_y$ and $R(b|y)=0$ otherwise, it is very easy to check that
%
%
\begin{align*}
\sum_{y=1}^{N}\mathbb{E}\sup_{b=1,\cdots,K-1}\Big|\sum_{x,a=1}^{N,K}\Big(M_{x,y}^{a,b}-M_{x,y}^{a,K}\Big)\sum_{i=1}^ng_iu_x^a(i)\Big|
=\mathbb{E}\sum_{x,y;a,b=1}^{N,K}M_{x,y}^{a,b}\big(\sum_{i=1}^ng_iu_x^a(i)\big)R(b|y).
\end{align*}
Denoting $Q(a|x)=\frac{1}{\|(g_i)_i\|_2}\sum_{i=1}^ng_iu_x^a(i)$ for every $x,a$\footnote{Actually, we should define $Q_\omega(a|x)$, where $g_i=g_i(\omega)$, for every $\omega$. However, the upper bounds below hold for every $\omega$, so we avoid that notation for simplicity.}, conditions (\ref{restrictions}) and Lemma \ref{preliminar lemma} tell us that $Q=\lambda P_1+\beta P_2$ and $R=\gamma P_3 + \delta P_4$, where $P_i\in S(N,K)$ for $i=1,\cdots, 4$, $|\lambda|+|\beta|\leq 1$ and $|\gamma|+|\delta|\leq 2$. The fact that $\omega(M)\leq 1$ guarantees that
\begin{align}\label{first term}
K_{1,2}\mathbb{E}\sum_{x,y;a,b=1}^{N,K}M_{x,y}^{a,b}\Big(\sum_{i=1}^ng_iu_x^a(i)\Big)R(b|y)=K_{1,2}\mathbb{E}\|(g_i)\|_2\sum_{x,y;a,b=1}^{N,K}M_{x,y}^{a,b}Q(a|x)R(b|y)\leq
2K_{1,2}\sqrt{n}.
\end{align}
On the other hand,
\begin{align*}
\ell(u_{N+1})\leq
K_{1,2}\mathbb{E}\Big|\sum_{x,y;a=1}^{N,K}M_{x,y}^{a,K}\sum_{i=1}^ng_iu_x^a(i)\Big|=K_{1,2}\mathbb{E}\|(g_i)\|_2\Big|\sum_{x,y;a=1}^{N,K}M_{x,y}^{a,b}Q(a|x)S(b|y)\Big|,
\end{align*}where $Q$ is defined as above and for every $y=1,\cdots, N$ we define $S(b|y)=1$ if $b=K$ and $S(b|y)=0$ otherwise. Again, $\omega(M)\leq 1$ implies that
\begin{align}\label{second term}
\ell(u_{N+1})\leq
K_{1,2}\mathbb{E}\|(g_i)\|_2\Big|\sum_{x,y;a=1}^{N,K}M_{x,y}^{a,b}Q(a|x)S(b|y)\Big|\leq
K_{1,2}\sqrt{n}.
\end{align}
Then, Equation (\ref{final estimate}) follows from Equations (\ref{first term}) and (\ref{second term}).
\end{proof}
\subsection{Some comments about the optimality}\label{Some comments about the optimality}
Lemma \ref{concentration} is not only optimal in the sense of Remark \ref{optimality I}, but it
can also be used to give a very simple proof of the optimal estimate $\Omega(\sqrt{\ln n})$ in Theorem \ref{optimal complementation}.
\begin{proof}[Proof of the first part of Theorem \ref{optimal complementation}]
To prove the first part of the statement let us consider linear maps $S:\ell_2^n\rightarrow \ell_1(\ell_\infty)$ and $T:\ell_1(\ell_\infty)\rightarrow \ell_2^n$ such that $T\circ S=id_{\ell_2^n}$. Then, we can realize $S=(S_x)_{x=1}^\infty$ and $T=(T_x)_{x=1}^\infty$ such that the linear maps $S_x:\ell_2^n\rightarrow \ell_\infty$ and $T_x:\ell_\infty\rightarrow \ell_2^n$ are defined by $S(z)=(S_x(z))_{x=1}^\infty$ for every $z\in \ell_2^n$ and $T((y_x)_{x=1}^\infty)=\sum_xT_x(y_x)$ for every $(y_x)_{x=1}^\infty\in \ell_1(\ell_\infty)$. Note that we easily have $\|T_x\|\leq \|T\|$ for every $x$. On the other hand, one can also check that
\begin{align*}
tr(T\circ S)= \sum_{x=1}^\infty tr(T_x\circ S_x).
\end{align*}Furthermore, according to Equation (\ref{Kahane-Khinchin}) we have
\begin{align*}
\sum_x\ell(S_x)\leq K_{1,2} \sum_x\mathbb{E}\Big\|\sum_{i=1}^ng_iS_x(|i\rangle)\Big\|_{\ell_\infty}=K_{1,2}\mathbb{E}\sum_x\Big\|S_x\Big(\sum_{i=1}^ng_i|i\rangle\Big)\Big\|_{\ell_\infty}\\
=K_{1,2}\mathbb{E}\Big\|S\Big(\sum_{i=1}^ng_i|i\rangle\Big)\Big\|_{\ell_1(\ell_\infty)}=K_{1,2}\mathbb{E}\Big\|\sum_{i=1}^ng_iS(|i\rangle)\Big\|_{\ell_1(\ell_\infty)}\leq
K_{1,2}\ell(S).
\end{align*}
Therefore, we have
\begin{align*}
n=tr(id_{\ell_2^n})=tr(T\circ S)= \sum_{x=1}^\infty tr(T_x\circ S_x)\leq C\frac{\sqrt{n}}{\sqrt{\ln n}}\sum_{x=1}^\infty\|T_x\|\ell(S_x)\\\leq
C\frac{\sqrt{n}}{\sqrt{\ln n}}\|T\|\sum_{x=1}^\infty\ell(S_x)\leq C'\frac{\sqrt{n}}{\sqrt{\ln n}}\|T\|\ell(S)\leq C'\frac{n}{\sqrt{\ln n}}\|T\|\|S\|,
\end{align*}where for the first inequality we have used Lemma \ref{concentration} and the last inequality follows from Equation (\ref{l-norm T}). The first part of the statement follows now trivially.

The second part of statement follows from well known results (see \cite{FLM} or \cite{JP}).
\end{proof}
\begin{remark}
One can easily verify that the previous proof also works if one considers complex Banach spaces. It is very interesting that in this case the $\sqrt{\ln n}$ factor  is optimal even in the noncommutative sense stated in the second part of Theorem \ref{optimal complementation}. Indeed, in the complex case one can define some operator space structures on the spaces $\ell_2^n$,  $\ell_1(\ell_\infty)$. We refer \cite{Pisierbook} for an introduction on the theory of operator spaces. In particular, we can consider the $R\cap C$ operator space structure on $\ell_2^n$ and the operator space structure on  $\ell_1(\ell_\infty)$ defined by the (operator space) projective tensor norm $\ell_1\hat{\otimes} \ell_\infty$ (this is usually referred as the natural operator space structure on $\ell_1( \ell_\infty)$). Then, the second part of Theorem \ref{optimal complementation} follows from \cite[Theorem 9]{JP}, which shows the existence of linear maps $j:R_n\cap C_n\longrightarrow \ell_1(\ell_\infty)$ and $P:\ell_1(\ell_\infty)\rightarrow R_n\cap C_n$ such that $P\circ j=id_{\ell_2^n}$ and $\|j\|_{cb}\|P\|_{cb}\leq \tilde{K}\sqrt{\ln n}$, where $\tilde{K}$ is a universal constant. 
\end{remark}
Let us finish this work with a final comment about the optimality of our main result Theorem \ref{Main II}. It can be deduced from \cite{JP} that for some Bell inequalities $M\in \mathcal M^{n,n}$ we have
\begin{align*}
\overline{\omega}_{OP_n}(M)\geq k\frac{n}{\ln n}\omega(M)
\end{align*}for some universal constant $k$. We do not know whether one can get the upper bound $O(\frac{n}{\ln n})$ in Theorem \ref{Main II}. On the other hand, $\overline{\omega}_{OP_n}(M)$
can be much larger than $LV_{|\psi_n\rangle}(M)$ for some $M\in
\mathcal M^{N,K}$. A particularly extreme case can be found for
$N=1$ and $K=n$, where one can find a certain element $M$ such that
$\omega(M)=\omega^*(M)=1$ and $\overline{\omega}_{OP_n}(M)\geq
\sqrt{n}$. Thus, another approach more focused on the specific
properties of the maximally entangled state could give a better
upper bound in Theorem \ref{Main I} without improving Theorem
\ref{Main II}. On the other hand, note that Theorem \ref{Theorem BRSW I} says that Theorem \ref{Main I} is very tight.

In the following, we will explain that in order to obtain an improvement of our Theorem \ref{Main II} a different approach from the one followed in this work is required. First, let us explain that Theorem \ref{Main II} and Theorem \ref{Main III} can be stated in terms of the so called $\gamma_2^*$ tensor norm. Given two Banach spaces $X$, $Y$ and their algebraic tensor product $X\otimes Y$, for a given $z\in X\otimes Y$ we define
\begin{align*}
\|z\|_{\gamma_{2,n}^*}=\sup\big\{\big\|(u\otimes v)(z)\big\|_{\ell_2^n\otimes_{\pi} \ell_2^n}\big\},
\end{align*}where the supremum runs over all linear maps $u:X\rightarrow \ell_2^n$, $v:Y\rightarrow \ell_2^n$ verifying $\|u\|, \|v\|\leq 1$. It is very easy to see that for every $M\in \R^{N^2K^2}$ we have
\begin{align*}
\omega_{OP_n}(M)=\|M\|_{\ell_1^N(\ell_\infty^K)\otimes_{\gamma_{2,n}^*}\ell_1^N(\ell_\infty^K)}.
\end{align*} Then, Theorem \ref{Main III} is equivalent to 
\begin{align}\label{gamma1}
\big\|id\otimes id:\ell_1^N(\ell_\infty^K)\otimes_\epsilon\ell_1^N(\ell_\infty^K)\rightarrow  \ell_1^N(\ell_\infty^K)\otimes_{\gamma_{2,n}^*}\ell_1^N(\ell_\infty^K)\big\|\leq D\frac{n}{\sqrt{\ln n}}.
\end{align}
To deal with general Bell inequalities one must replace the space $\ell_1^N(\ell_\infty^K)$ with the space $NSG^*(N,K)$ introduced in \cite[Section 5]{JP}. Then, Theorem \ref{Main II} is equivalent to 
\begin{align}\label{gamma2}
\big\|id\otimes id:NSG^*(N,K)\otimes_\epsilon NSG^*(N,K)\rightarrow  NSG^*(N,K)\otimes_{\gamma_{2,n}^*}NSG^*(N,K)\big\|\leq D'\frac{n}{\sqrt{\ln n}}.
\end{align}
It was proven in \cite[Section 5]{JP} that $NSG^*(N,K)$ is a \emph{twisted version} of the space $\ell_1^N(\ell_\infty^K)$ and it can be seen that proving (\ref{gamma1}) is equivalent to prove (\ref{gamma2}) with a slight modification in the constant. 

A careful study of the proof of Theorem \ref{Main
II} presented before shows that we have actually reduced the problem to study the
picture in which we have two linear maps $S:\ell_2^n\rightarrow
\ell_1(\ell_\infty)$ and $T:\ell_1(\ell_\infty)\rightarrow \ell_2^n$
such that $T\circ S=id_{\ell_2^n}$ and we must study how small
$\|T\|\|S\|$ can be; that is, the best complementation constant of $\ell_2^n$ in $\ell_1(\ell_\infty)$. Indeed, this estimate perfectly fits in the picture explained above and it can be used to upper bound the norms in (\ref{gamma1}) and (\ref{gamma2}). Hence, the fact that the estimate provided in the first part of Theorem \ref{optimal complementation} is optimal means that we cannot get a better upper bound in Theorem \ref{Main II} by reducing the problem in the way we have done in this work. 
\section*{Acknowledgments}
We would like to thank M. Junge, O. Regev and T. Vidick for many helpful discussions on previous versions.

Author's research was supported by EU grant QUEVADIS, Spanish projets QUITEMAD, MTM2011-26912 and MINECO: ICMAT Severo Ochoa project SEV-2011-0087 and the ``Juan de la Cierva'' program.

\vskip 2cm

\hfill \noindent \textbf{Carlos Palazuelos} \\
\null \hfill Instituto de Ciencias Matemáticas\\ \null \hfill
CSIC-UAM-UC3M-UCM \\ \null \hfill Consejo Superior de
Investigaciones Cient{\'\i}ficas \\ \null \hfill C/ Nicol\'as Cabrera 13-15.
28049, Madrid. Spain \\ \null \hfill\texttt{carlospalazuelos@icmat.es}

\end{document}